\documentclass{article} 
\usepackage{nips11submit_e, url, amsmath, amsfonts, amssymb, cite, graphicx, caption, enumerate, stmaryrd, subfigure, amsthm, enumerate}

\newcommand{\hide}[1]{}

\newenvironment{thmtag}[2][Theorem]{\begin{trivlist}
\item[\hskip \labelsep {\bfseries #1}\hskip \labelsep {\bfseries #2}]}{\end{trivlist}}

\newcommand{\eod}{{${}$\\}}
\newcommand{\cA}{{\mathcal A}}

\newcommand{\ccH}{{\mathcal H}}
\newcommand{\cN}{{\mathcal N}}
\newcommand{\cS}{{\mathcal S}}

\newcommand{\cX}{{\mathcal X}}

\newcommand{\bE}{{\mathbb E}}
\newcommand{\bI}{{\mathbb I}}
\newcommand{\bR}{{\mathbb R}}
\newcommand{\reward}{{R}}

\newtheorem{thm}{Theorem}

\newtheorem{assm}{Assumption}
\newtheorem{eg}{Example}
\newtheorem{defn}{Definition}

\theoremstyle{remark}
\newtheorem{rem}{Remark}

\title{Metabolic cost as an organizing principle\\
for cooperative learning}

\author{
David Balduzzi$^1$, Pedro A. Ortega$^2$, Michel Besserve$^2$\\
$^1$Swiss Federal Institute for Technology, ETH Z{\"u}rich \\
$^2$Max Planck Institute for Intelligent Systems, T{\"u}bingen, Germany \\
\texttt{david.balduzzi@inf.ethz.ch, \{ortega, besserve\}@tuebingen.mpg.de}
}

\nipsfinalcopy 
\begin{document}

\maketitle

\begin{abstract}
This paper investigates how neurons can use metabolic cost to facilitate learning at a population level. Although decision-making by individual neurons has been extensively studied, questions regarding how neurons should behave to cooperate effectively remain largely unaddressed. Under assumptions that capture a few basic features of cortical neurons, we show that constraining reward maximization by metabolic cost aligns the information content of actions with their expected reward. Thus, metabolic cost provides a mechanism whereby neurons encode expected reward into their outputs. Further, aside from reducing energy expenditures, imposing a tight metabolic constraint also increases the accuracy of empirical estimates of rewards, increasing the robustness of distributed learning. Finally, we present two implementations of metabolically constrained learning that confirm our theoretical finding. These results suggest that metabolic cost may be an organizing principle underlying the neural code, and may also provide a useful guide to the design and analysis of other cooperating populations.
\end{abstract}

\section{Introduction}

Rational decision making is typically formalized as optimizing a reward function \cite{sutton:98, Russel:09, OrtegaBraun:11}. This paper investigates neuronal learning from an optimization perspective. We assume that both the brain as a whole and individual neurons are rational decision makers optimizing reward functions of some kind. 

Since the brain learns from finite samples, it is exposed to a tradeoff between over- and under-fitting: increasing a model's capacity can improve its fit on training data, but potentially worsens performance on future samples \cite{vapnik:82}. Remarkably, however, the human brain effortlessly handles a wide-range of complex pattern recognition tasks suggesting it both has a large capacity and, paradoxically, also generalizes extremely well. On the basis of these conflicting observations, it has been argued that \emph{useful biases} in the form of ``generic mechanisms for representation'' must be hardwired into cortex \cite{geman:92}. Our goal in this paper is to propose a bias that is both useful and biologically plausible.

Let us outline the problem. Neurons learn inductively. They can generalize from finite samples and encode estimates of future outcomes (for example, rewards) into their spiketrains \cite{Gottfried:2003fk}. Learning-theoretical results imply that generalizing successfully from small samples requires strong biases \cite{vapnik:82} or, in other words, specialization. Thus, at any given time some neurons' specialties are more relevant than others. Since most of the data neurons receive are other neurons' outputs, it is essential that they indicate which of their outputs encode high quality estimates. Downstream neurons should then be biased to specialize on these outputs, thereby reducing the effective search space that neurons explore.

The question we ask is: \emph{How can a} single \emph{neuron near-maximize its} expected \emph{reward and simultaneously help downstream neurons do the same?} 

As a partial answer, we propose the following organizing principle. Neurons should consistently label outputs as useful and not useful. Specifically, the optimizations performed by neurons should be designed so that spikes are useful and silences are not. For example, it is well known that, on average, the more spikes a neuron receives, the more it learns by modifying its synapses. Although this fact is often taken for granted, it requires explanation. We suggest that spikes drive learning because they are useful. By useful, we mean an output that  (i) predicts high reward with (ii) tight confidence intervals.

Distinguishing useful from irrelevant outputs is helpful to downstream neurons since it reduces the size of the search space they are confronted with. If spikes reliably predict future reward, then neurons should be biased towards learning from spikes. In fact, this fits well with experimental evidence \cite{dan:04, dan:06, pawlak:10}. Moreover, consistently biasing learning toward outputs labeled as useful (i.e. spikes) also provides a principled way to reduce capacity, and thus improve generalization guarantees, without sacrificing empirical performance.

The scope of this paper is limited to showing how neurons may systematically distinguish useful from irrelevant outputs. Fleshing out the implications for learning at the population level is deferred to future work. Furthermore, we only consider excitatory connections in this paper; the role of inhibition is also deferred to future work.

\paragraph{Overview.}
We explore the consequences of a few basic assumptions regarding spikes, metabolic cost and neuronal rewards, see \S\ref{s:setup}.  We consider a minimal model of biological neural network from which we extract guiding principles that may apply, at least approximately, to more complex, biologically realistic models. The most important assumptions are that neurons aim to maximize reward after spiking, and that neurons have to operate within a fixed metabolic budget that constrains how often they can spike in a given time interval.

It turns out that these seemingly innocuous assumptions are the key to distinguishing useful from irrelevant outputs. Our first result is that constrained reward maximization causes neurons to maximize the information they encode into their spikes, see \S\ref{s:inf-spike}. Moreover, if spikes are sufficiently rare, it turns out that spikes dominate the information communicated by neurons, which has interesting implications for credit assignment in cortex, see \S\ref{s:inf-spike} and \cite{bt:12}.

If neurons attempt to maximize empirical reward \emph{after spiking}, it follows that neurons encode reward estimates into spikes. When a neuron produces a spike, it thus signals to downstream neurons that it expects a positive neuromodulatory signal such as dopamine. This raises questions concerning the quality of the reward estimates encoded in spikes, see \S\ref{s:eve}. We show that the more information neurons encode into spikes, the tighter the guarantees tying empirical reward to expected reward. 

We conclude by describing some implications of our results for cooperative optimization. Theorems are proved in the appendices.

\paragraph{Related work.}
Neuronal plasticity and its implications for the neural code have been intensively studied for many years. The work closest in spirit to this paper is Seung's ``hedonistic'' synapses, which seek to increase average reward \cite{seung:03}. A second related line of research applies the information bottleneck method -- an alternate constraint to the one considered here --  to neuronal learning \cite{tishby:99, buesing:07}. 

An information-theoretic perspective on synaptic homeostasis that complements the results in this paper is \cite{bt:12}. The consequences of constrained optimization for spike-timing dependent plasticity (STDP) are presented in \cite{bb:12}. A practical implementation of regularized STDP, inspired by the ideas presented here, can be found in \cite{Nere:2012fk}.

\paragraph{Acknowledgements.}
We thank Yevgeny Seldin and Giulio Tononi for useful discussions.

\section{A minimal model}
\label{s:setup}

The mammalian cortex contains between $10^7$ and $10^{11}$ neurons (depending on the species) that are guided by neuromodulators signaling pleasure, pain and other globally salient events. Neurons communicate with each other via spiketrains -- sequences of silences and spikes they receive from and transmit to $10^3$ to $10^4$ other neurons through connections called synapses. Neurons learn by increasing and decreasing the efficacy of their synapses according to the timing of pre-synaptic (input) and post-synaptic (output) spikes, as well as the presence or absence of neuromodulators  \cite{Markram:1997fk, song:00, dan:04, dan:06, pawlak:10}.

\paragraph{Model neurons.}
The cortex can be modeled as a population of $K$ neurons. Neuron $n_k$ follows policy $\pi_k$, where $\pi_k(a|s)$ is a Markov matrix specifying the probability the neuron picks action, or output, $a \in \cA=\{0,1\}$ upon encountering situation $s \in \cS$. A situation, or input, is a vector $s=(a_1, \ldots, a_K)$ whose entries are the actions of all the neurons in the brain at the previous time step. Since each neuron is only exposed to a small fraction of the brain, it ignores most entries in the vector. Thus, the policy of neuron $k$ is $\pi_k(a|s) = \pi_k(a|\Sigma_k)$ for some subset $\Sigma_k \subset \{a_1, \ldots, a_K\}$. Let $P(S)$ denote the prior over situations. We will assume the prior is i.i.d. when providing guarantees on estimates, see Remark~\ref{r:iid} for a brief discussion.

Neurons are exposed to a global \emph{neuromodulatory signal} $\nu \in \cN$ that signals the performance of the population as a whole. Neuromodulators are drawn with probability $P(\nu|s)$.

We make the following assumptions.

\begin{assm}[reward maximization]\label{a:max}\eod
	Neurons maximize a reward function that depends on neuromodulatory signals, input spikes and output spikes:
	\begin{equation}
		\widehat\pi_k = \underset{\pi\in {\mathcal M}}{\arg\max} \sum_{i=1}^N \reward_k(s_i,a_i,\nu_i),
	\end{equation}
	where $a_i$ is the output chosen by $\pi$ in response to input $s_i$.	
\end{assm}

The set ${\mathcal M}$ is the set of possible neuronal mechanisms, a subset of the set Markov matrices on $S\times A$. Two examples of ${\mathcal M}$ that are relevant to our discussion are discrete threshold neurons
\begin{equation*}
	{\mathcal M} = \left\{H(\langle\mathbf{w},\mathbf{s}\rangle - \vartheta)\,\big|\, \mathbf{w}\in \bR^K\right\}
\end{equation*}
where $H(\bullet)$ is the Heaviside function, see \cite{bb:12}, and the full set of Markov matrices, see discussion of $Q$-learning below.

Note that, since different neurons are exposed to different subsets of the total brain activity, they have different reward structures, even for the same neuromodulators: in general $\reward_k\neq \reward_j$ for $j\neq k$ since $\Sigma_k\neq \Sigma_j$.

Since we focus on the behavior of a single neuron, we will often drop the subscript $k$ from the notation below.

\begin{assm}[spikes gate rewards]\label{a:gate}\eod
	The reward function is gated by synaptic outputs:
	\begin{equation}
		\reward(s,a,\nu) = \reward(s,\nu)\cdot \bI_{a=1},
		\quad \text{ where }\bI_{a=1}= \begin{cases}
			1 & a=1\\
			0 & \text{else}
		\end{cases}
		\quad
		\text{ is the indicator function.}
	\end{equation}	
\end{assm}

Neurophysiological evidence suggests that neurons only potentiate or depotentiate their synapses shortly before or after producing spikes \cite{Markram:1997fk, song:00, dan:04, dan:06}. This encourages specialization: neurons search for a small set of inputs that reliably predict future reward signals -- other inputs are ignored. 

For simplicity we assume neurons have two outputs, spikes and silence. We use notations $a_0$ or $a=0$ for silence and $a_1$ or $a=1$ for spikes.

\begin{assm}[metabolic budget]\label{a:metabolic}\eod
	Neurons have a fixed metabolic budget that determines the maximum frequency of spiking over some (unspecified) time period:
	\begin{equation}
		\pi(a_1) \leq \rho
	\end{equation}
	where $\pi(a_1)=\sum_s p(s)\cdot \pi(a_1|s)$ is the spiking frequency under policy $\pi$.
\end{assm}

Although a soft constraint is more biologically plausible, imposing a hard constraint simplifies the exposition without significantly altering the conclusions. The main effect of softening the constraints is to allow the information carried by spikes and the capacity of neurons to vary, thereby softening Theorems~\ref{t:conc-inf} and \ref{t:vcr} below.

 Since neurons only modify their synapses when they spike, it follows that neurons that spike very infrequently learn very little if at all. Thus, we expect that there are mechanisms in place ensuring that neurons not only stay within their metabolic budget, but also that they come close to using all of it.
 
Spikes and silences are not abstract, interchangeable symbols. Spiking and responding to spikes carries a much higher metabolic cost than not spiking \cite{Hasenstaub:2010fk}. This cost is significant since the nervous system consumes a disproportionate share of an organism's total energy budget \cite{attwell:01}. It has been hypothesized that a function of sleep is to homeostatically regulate synaptic strengths, so as to control metabolic expenditures associated with action potentials, see \cite{tononi:03b, Gilestro:2009fk, vyazovskiy:09, Maret:2011rt}.

Unlike ``agents in the wild'', individual neurons have negligible impact on what happens next: a single neuron has little effect on the neurons it targets -- since each receives inputs from thousands of other neurons. The inability of individual neurons to manipulate their environment simplifies the optimization problems they face by stripping out the recursive Bellman aspect \cite{bellman:57}:

\begin{assm}[disempowerment]\label{a:disempower}\eod
	An individual neuron has no immediate influence on its environment
	\begin{equation}
		\label{e:disempower}
		p\Big(s(t+1,\ldots)\,\Big|\,s(t),a_k(t)\Big) = p\Big(s(t+1,\ldots)\,\Big|\,s(t)\Big),
	\end{equation}
	where $t$ refers to time. I.e. conditioning on the neuron's output makes no difference to the distribution on subsequent actions by the rest of the brain.
\end{assm}

The future situations a neuron encounters are essentially unaffected by its output. This assumption fails at the population level -- populations of neurons necessarily affect the organism's actions. Nevertheless Eq.~\eqref{e:disempower} is a reasonable assumption at the individual neuron level since, for example, destroying a single neuron makes essentially no difference to brain function.

\section{Encoding useful information in spikes}
\label{s:inf-spike}

This section considers the implications of our assumptions for the information content of spikes. Mutual information provides a formal method for quantifying the information content of a channel. We present a related measure, effective information, that measures the information content of a single output in terms of how much that output reduces prior uncertainty regarding the set of inputs. Importantly, we show that under the above assumptions, spikes not only reduce uncertainty, but also signal that the neuron received an input that was historically followed by high rewards. The information encoded in spikes is thus indicative of expected future rewards.

\paragraph{Information.}
It is useful to consider neurons as \emph{communication channels} mapping situations to actions. The average information communicated by a neuron is then the mutual information $I_\pi(\cS,\cA)$. However, since we are interested in the information communicated by specific actions (in particular, spikes), we introduce effective information\footnote{We extend the definition of effective information in \cite{bt:08} to allow arbitrary priors instead of restricting to the uniform distribution.} \cite{bt:08, bt:09} which quantifies the information that a single output encodes about the input.

\begin{defn}[effective information]\eod
	\label{d:ei}
	Given a neuron with policy $\pi(a|s)$ and prior $P(\cS)$ on situations, the \emph{effective information} generated by action $a\in\cA$ is
	\begin{equation}
		ei(\pi,a) := D\big[\pi(S|a)\,\big\|\,P(S)\big]
		\,\,\,\text{ where } \pi(s|a):=\frac{\pi(a|s)}{\pi(a)}\cdot P(s)
	\end{equation}
	is computed via Bayes' rule and $D[\bullet\|\bullet]$ is the Kullback-Leibler divergence $D[p\|q]:=\sum p_i\log\frac{p_i}{q_i}$.
\end{defn}

An interesting special case is when the prior on situations is the uniform distribution and the policy is deterministic. It follows that
\begin{equation}
	\label{e:selectivity}
	ei(\pi,a) = -\log \frac{|\pi^{-1}(a)|}{|\cS|},
\end{equation}
where $|\bullet|$ denotes cardinality and (since the policy is deterministic), $\pi$ is  a function $\pi:\cS\rightarrow\cA$. In Eq~\eqref{e:selectivity},  effective information quantifies the \emph{selectivity} of an output: the fraction of inputs causing the policy to output $a$. The smaller the fraction, or alternatively the more sensitive output $a$ is to perturbations in the input, the higher effective information \cite{bt:12}.

\begin{rem}
	Suppose we have model $P_{\mathcal M}(d|h)$ that specifies the probability of observing data given a hypothesis. Further suppose we have prior distribution $P(h)$ on hypotheses. If we observe data $d$, how much have we learned about the hypotheses? The \emph{Bayesian information gain} is 
	\begin{equation}
	\label{e:big}
	D\big[P_{\mathcal M}(H|d)\,\big\|\,P(H)\big].
	\end{equation}
	If we consider a neuron's policy as a model, with inputs as hypotheses and outputs as evidence, then effective information quantifies the Bayesian information gained about the inputs given an output.
\end{rem}

\begin{rem}
	The \emph{expectation} of $ei$ is mutual information: $\bE_{\pi(a)}\big[ei(\pi,a)\big]=I_\pi(\cS,\cA)$.	
\end{rem}

\paragraph{Information aligns with rewards (theory).}
We show that neurons implementing constrained reward maximization from Assumption~\ref{a:max} also maximize the effective information of their spikes $ei(\pi,a_1)$, that we will call \emph{information per spike}.

\begin{defn}[empirical reward]\eod
	Given a finite sample of situations, actions and neuromodulators $(s_i,a_i,\nu_i)_{i=1}^N$, let the \emph{empirical reward} observed after performing action $a$ in situation $s$ be
	\begin{equation*}
		\widehat{\reward}(s,a):=\frac{1}{N}\sum_{\{i|s_i=s,a_i=a\}} \reward(s_i,a_i,\nu_i).		
	\end{equation*}	
	We also introduce the \emph{empirical reward after spiking}, $\frac{1}{N} \sum_{\{i|a_i=1\}} \reward\big(s_i,a_1,\nu_i\big)$ and the 
	\emph{empirical reward per spike} $\frac{1}{|\{i|a_i=1\}|} \sum_{\{i|a_i=1\}} \reward\big(s_i,a_1,\nu_i\big)$.
\end{defn}

\begin{thm}[maximizing reward/spike maximizes information/spike]\eod
	\label{t:maximizer}
	Assume that situations yield different empirical rewards, i.e. $\widehat{R}(s,a_1)\neq \widehat{R}(s',a_1)$ for all $s,s'\in\cS$. Recall that by Assumptions~\ref{a:max} and \ref{a:metabolic} an optimal policy $\widehat\pi$ satisfies
	\begin{equation*}
		\widehat{\pi} = \underset{\{\pi\in {\mathcal M}|\pi(a_1)\leq \rho\}}{\arg\max} \sum_{i=1}^N \reward(s_i,a_i,\nu_i), 
	\end{equation*}
	If we also apply constraint $\pi(a_1)\geq\rho$, i.e. $\rho$ is both upper and lower bound, then the optimal policy $\widehat{\pi}$ maximizes information per spike. More precisely, the optimal policy $\widehat{\pi}$ satisfies 
	\begin{equation*}
		ei(\widehat{\pi},a_1)\geq ei(\pi,a_1) \text{ for all }\pi \text{ such that }\pi(a_1)=\rho.
	\end{equation*}
\end{thm}

\begin{rem}
	The optimal policy will satisfy $\widehat{\pi}(a_1)=\rho$ if there are enough inputs $s$ satisfying $\bE_\nu[R(s,\nu,a_1)]>0$. In other words, if there are enough situations where spiking, on average, is followed by positive reward. Imposing the lower bound in the theorem means we compare the optimal policy with alternate policies that spike with the same frequency.
\end{rem}

Thus, the optimal policy necessarily maximizes both the empirical reward after spiking and the effective information encoded in spikes. We illustrate this result with simulations below.

\paragraph{Information aligns with rewards (experiments).}
A learning algorithm implementing constrained reward maximization is a modification of $Q$-learning \cite{watkins:92}:

\begin{eg}[Metabolically constrained $Q$-learning]\eod
	\label{eg:mq}
	If a neuron chooses action $a$ in situation $s$ and subsequently receives neuromodulator $\nu$, then it updates the $Q$-matrix by
	\begin{equation*}
		Q(s,a)\leftarrow Q(s,a)+ \alpha\cdot\Big[\widehat{\reward}(s,a,\nu)-Q(s,a)\Big],
	\end{equation*}
	where $\alpha$ controls the rate. After updating $Q$, the neuron constructs new policy
	\begin{equation*}
		\pi(a|s) = {\mathcal M}^n\Big(e^{Q(s,a)}\Big).
	\end{equation*}
	Operation ${\mathcal M}(\bullet)$ renormalizes the policy twice: first by $Z(a)$ chosen such that $\sum_{s\in\cS}\pi(a|s)P(s)=P(a)$ for all $a$, and then by $Z(s)$ chosen such that $\sum_{a\in\cA}\pi(a|s)=1$ for all $s$. Setting $n=3$ yields a policy that approximately implements the metabolic constraint.
\end{eg}

Figure~\ref{f:ei-ev} shows how effective information and empirical reward after spiking covary as neurons $Q$-learn. We initialized 5000 neurons randomly and applied metabolic constraints $\rho\in\{0.1,0.3,0.5\}$. Rewards are drawn randomly. As the neurons adapt, their policies become both more deterministic and more likely to spike in situations yielding higher rewards, so as neurons adapt they both encode more information into their spikes and predict higher rewards after spiking. The tighter the metabolic constraint (i.e. the lower $\rho$), the higher the empirical reward after spiking. Thus, the information encoded in spikes provides a reliable guide to the empirical reward after spiking.

\paragraph{Spikes dominate information content.}
Theorem~\ref{t:maximizer} shows information per spike is maximized by constrained reward maximization. Theorem~\ref{t:conc-inf} below consider how much of the total information communicated by a neuron is carried by spikes.

\begin{thm}[spikes carry essentially all information]\eod
	\label{t:conc-inf}
	Suppose a neuron has two actions (silence $a_0$ and spike $a_1$) and produces spikes infrequently: $\pi(a_1)\ll1$. Then the total information communicated by the neuron is approximately the information it communicates using spikes alone:
	\begin{equation}
		I_\pi(\cS;\cA) = \pi(a_1)\cdot ei(\pi,a_1)
		+ O\left(\pi(a_1)^2\right).
	\end{equation}
\end{thm}

If the metabolic constraint is tight, meaning $\rho$ and so $\pi(a_1)$ are small, then for the optimal policy encodes a lot of information into spikes. In this setting, Theorem~\ref{t:conc-inf} implies that the information communicated by a neuron is (up to first order) carried by spikes alone.

Eq~\eqref{e:selectivity} and  Theorem~\ref{t:conc-inf} together have interesting implications for credit assignment, see \cite{bt:12} for details. In particular, if spikes carry most of the information in cortex, then spiking neurons and synapses should reinforced in response to positive global signals such as dopamine, and conversely for negative global signals. Neurons and synapses that are silent contribute little to the information generated by cortex, and so should be neither potentiated nor depotentiated. This fits neurophysiological evidence suggesting that spikes play a distinguished role in synaptic potentiation and depotentiation \cite{Markram:1997fk, dan:04, dan:06, pawlak:10}.

\paragraph{Encoding reward estimates in spikes.}
Finally, we briefly illustrate the effect of the metabolic constraint $\rho$ on the empirical reward per spike. Suppose we have sampled empirical rewards $\big\{\widehat{\reward}(s,a_1)\big|s\in\cS\big\}$. The optimal policy is constructed as follows. First, rank states by their empirical reward. Let $\cS_\rho$ denote the states in the top $\rho^{th}$ percentile. Define 
\begin{equation*}			
	\widehat{\pi}(a_1|s)=
	\begin{cases}
		1 & s\in \cS_{\rho}\\
		0 & \text{else.}
	\end{cases}
\end{equation*}

The optimal policy can be visualized as moving a window of fixed size and variable shape over the input space, such that the underlying configuration maximizes reward. 

Figure~\ref{f:bfo} shows an example. Situations are ranked according to their empirical reward. The metabolic constraint is set at $\rho=15\%$, so the optimal policy spikes for the 15\% of situations with highest reward. The policy picking these situations results in a reward per spike of 0.89. Situations which do not cause spikes receive \emph{no reward}; if spikes and silences would be exchanged, this reward would be .21 (we call it reward after silence).

We make two observations. First, tightening the metabolic constraint, so that the policy spikes for $\ll 15\%$ of situations, increases the reward per spike. Second, the variance in reward per spike is much lower than after not spiking, and typically decreases with $\rho$. 

As a general rule of thumb, tightening the metabolic constraint by decreasing $\rho$ both increases the empirical reward per spike and reduces the variance in empirical reward per spike.

\section{Guarantees on reward estimates}
\label{s:eve}

 In a dynamically changing environment like the cortex, it is important that neurons reliably represent high reward. This section shows that the reliability of spikes depends on how much information is encoded in them. 

\paragraph{Theoretical guarantees.}
We say that a neuron's spikes reliably represent reward when there is a low variability in the empirical mean reward when the neuron spikes. For a given policy $\pi$. The expected and empirical rewards per spike are
\begin{gather*}
	\reward_\pi:= \bE\Big[\reward(S,a_1,N)\Big] := \sum_{s,n}\pi(s|a_1)\cdot P(\nu|s)\cdot\reward(s,a_1,\nu)
	\text{ and}\\
	\widehat{\reward}_\pi:= \frac{1}{T_1}\sum_{\{(s_t,\nu_t)|\pi(s_t)=1\}} \reward\big(s_t,a_1,\nu_t\big)\text{ respectively},
\end{gather*}
where $T_1$ counts spikes produced by the neuron during $[1,T]$.  We measure sample size in terms of spikes (rather than spikes and silences) because spikes are metabolically expensive, so only spikes count towards the cost of collecting a finite sample.

Let us introduce notation for computations with respect to the \emph{uniform prior} $U(s)=\frac{1}{|\cS|}$, where $|\cS|$ is the total number of possible situations. Let $\pi_u(a)=\sum_{s}\pi(a|s)\cdot U(s)$, $\pi_u(s|a)=\pi(a|s)\frac{U(s)}{\pi_(a)}$, and  $ei_u(\pi,a)=\sum_s \pi_u(s|a)\log\frac{\pi_u(s|a)}{\pi_u(a)}$. Note that $ei_u$ recovers the original notion of effective information in \cite{bt:08}, which Definition~\ref{d:ei} generalizes.

The following theorem is proved using a version of Occam's razor \cite{seldin:08}.  

\begin{thm}[error bound for empirical reward]\eod
	\label{t:vcr}
	Suppose the neuron chooses a deterministic policy $\pi$ under the constraint that it spikes for a fixed fraction of situations: $\pi_u(a_1)=const$. Further, suppose that situations are sampled i.i.d. Without loss of generality,\footnote{Since only relative rewards affect the choice of optimal policy, it follows that negative rewards can be stripped out of the optimization problem by introducing an additive constant.} assume that rewards lie in $[0,b]$. Then with probability at least $1-\delta$,
	\begin{equation}
		{\left|\reward_\pi-\widehat{\reward}_\pi\right|} \leq
		b\cdot\sqrt{|\cS|\cdot\frac{ei_u(\pi,a_1)+1}{2T_1\cdot e^{ei_u(\pi,a_1)}}+\frac{\log\frac{2}{\delta}}{2T_1}}
	\end{equation}
	Guarantees improve as $ei_u$ increases since $\frac{x+1}{e^x}$ decreases as $x$ increases. 	
\end{thm}

Note that $ei$ and $ei_u$ covary since increasing the number of situations where a neuron spikes decreases \emph{both} $ei$ and $ei_u$; similarly, decreasing the number of situations where a neuron spikes increases both $ei$ and $ei_u$.

Thus, tightening the metabolic constraint in Assumption~\ref{a:metabolic}, by choosing low $\rho$, yields policies that have better guarantees on their reward estimates.

\begin{rem}
	\label{r:iid}
	The assumption that inputs are i.i.d. is not realistic for cortical neurons. We make two remarks. First, if rewards are only non-zero in the presence of neuromodulatory signals, then the assumption  states that situations directly preceding neuromodulator release are i.i.d, which is more reasonable. Second, similar results have recently been obtained in non-i.i.d. scenarios using more sophisticated PAC-Bayes methods \cite{rubin:11} -- and these may be applicable to our setting.
\end{rem}

\paragraph{Guarantees in practice.}
Figure~\ref{f:ei-bound} plots information encoded in spikes ($x$-axis) against the difference between the normalized empirical and expected reward ($y$-axis). Rewards were drawn randomly and the expected and empirical error for 16,000 deterministic policies with $|\cS|=50$, $P(\cS)$ uniform and $T_1=20$ were computed. Policies were sampled randomly with $k$, the number of situations causing the policy to spike, varying uniformly across $[1,25]$. The figure shows that both normalized error and the standard error of the error decrease as $ei$ increases. Figure~\ref{f:ei-bound} confirms that the bound in Theorem~\ref{t:vcr} is a reasonable guide to performance in practice.

Thus, the more information encoded in spikes, the better a neuron's empirical estimate of its expected reward per spike. The metabolic constraint in Assumption~\ref{a:metabolic} controls the quality of a neuron's empirical estimates of its expected reward.

\section{Discussion}
\label{s:discussion}

The space of possible policies that the cortex as a whole could implement is \emph{vast}: it consists in choosing synaptic weights of  millions or billions of neurons each receiving inputs from thousands or tens of thousands of other neurons. However, the space of policies that makes sense biologically is probably much smaller due to biological as well as learning-theoretic considerations. This paper has shown that metabolically constrained reward maximization provides a biologically plausible way for neurons to distinguish useful outputs from those that are not. It follows from our assumptions that spikes:
\begin{itemize}
	\item are responsible for most of the information communicated by neurons;
	\item signal when, based on empirical estimates, neurons predict high reward; and
	\item come equipped with performance guarantees that increase as the metabolic constraint is tightened, and so more information is encoded in spikes.
\end{itemize}

This suggests that neurons should \emph{privilege spikes} during learning; and indeed there is a large body of experimental evidence that this is exactly what occurs \cite{Markram:1997fk, song:00, dan:04, dan:06}: synaptic plasticity is triggered by pre- and post-synaptic spikes, with the decision to potentiate or depotentiate depending on their precise timing.

An important unanswered question is setting the metabolic constraint $\rho$. Clearly, if $\rho$ is too low, then neurons will barely fire at all, which is not desirable. Conversely, if $\rho$ is too high then information and empirical reward encoded in spikes, as well as the quality of the empirical estimates, all degrade, which is also to be avoided. In this paper, we have simply highlighted the metabolic constraint $\rho$ as an important lever that neurons may actively manipulate. Finding an optimal value or range of values for $\rho$ is outside the scope of this paper. 

\paragraph{Biasing neuronal mechanisms.}
$Q$-learning is not a practical learning rule; it simply constructs a giant lookup table. If we think of the set of inputs causing a neuron to fire as a window of varying shape, then the metabolic constraint fixes the size of the window and $Q$-learning places no other constraints on its shape.

However, if upstream neurons systematically encode information and reward in spikes, then it makes sense to bias shape of downstream neuronal ``firing windows'' to take the asymmetry between spikes and silences into account. This is exactly what we occurs in cortex. The vast majority of synapses are \emph{excitatory}: the more input spikes neurons receive, the more likely they are to spike themselves. 

Cortical neurons are biased toward firing for more spikes, which makes sense if spikes are reliable predictors of future reward. Thus, neuronal ``firing windows'' are shaped such that if a firing pattern causes a neuron to spike, then so does any firing pattern containing strictly more spikes. Moreover, sufficiently many pre-synaptic spikes will (essentially) always cause a neuron to fire. 

Thus, neurons aggregate evidence for high reward (spikes from upstream neurons), and modify their synaptic strengths to maximize their empirical reward. Fine-tuning is necessary since no two neurons have exactly the same connectivity, and therefore no two neurons use the same data to predict global neuromodulatory signals. 

Inhibitory (GABA) synapses do not fit this picture. Extending our framework for plasticity to include both types of neurons is challenging. However, we conjecture this incompatibility can be overcome, with inhibition playing a complementary role; possibly centered on imposing sparse activity in the brain \cite{Olshausen:1997xy, Olshausen:2004fk} and selecting competing neural assemblies and structures.

\paragraph{Global versus overlapping local optimizations.}
The results in this paper depend on specific assumptions and model choices. A particularly important assumption is that all neurons attempt to maximize the same global neuromodulatory reward signal. In this scenario, since neurons differ in their connectivity, they have access to different subsets of brain activity representing different environemental features and events, and therefore specialize on different sources of reward. 

The reality is more complicated with multiple overlapping neuromodulatory signals including dopamine, noradrenaline, acetylcholine and others. Moreover, neurons involved in, say, early visual processing may not require neuromodulatory guidance; rather, they may search for stable invariants over short time frames (hundreds of milliseconds). There is likely a diverse array of reward functions implemented across cortex. 

It is thus unclear whether the brain can be accurately described as optimizing a single well-defined reward function. Nevertheless, until decisively shown to be false, we believe the optimization perspective to be a fruitful working hypothesis. Even if it does not apply to the brain as a whole, it may nevertheless provide insight into how populations of neurons in specific brain areas converge on useful behaviors.

\paragraph{A spiking currency.}
Finally, it is interesting to speculate on an analogy between spikes and paper currency. Money plays many overlapping roles in an economy, including: {\rm (i)} focusing attention; {\rm (ii)} stimulating activity; and {\rm (iii)} providing a quantitative lingua franca for tracking revenues and expenditures. Note that, as for the brain, it is unclear whether an economy as a whole can be reduced to optimizing a single well-defined function.

Spikes may play similar roles in cortex to those of paper currency in an economy. Spikes focus attention: STDP and other proposed learning rules are particularly sensitive to spikes and spike timing. Spikes stimulate activity: input spikes cause output spikes. Finally, spikes leave trails of (Calcium) traces that are used to reinforce and discourage neuronal behaviors in response to neuromodulatory signals. 

Neither money nor spikes are intrinsically valuable. Currency can be devalued by inflation. Similarly, the information content and guarantees associated with spikes can be eroded by overpotentiating synapses which reduces their selectivity (potentially leading to epileptic seizures in extreme cases). Regulating the information content of spikes is therefore essential. Assumption~\ref{a:metabolic} provides a simple constraint that can be approximately imposed by regulating synaptic weights. Indeed, there is evidence that one of the functions of sleep is precisely this \cite{tononi:03b, Gilestro:2009fk, vyazovskiy:09, Maret:2011rt}. 

Spikes with high information content are valuable because they come with strong guarantees on their estimates. They are therefore worth paying attention to, worth responding to, worth keeping track of, and worth learning from. 

{
}

\setcounter{section}{0}
\renewcommand{\thesection}{A.\arabic{section}}

\section{Proof of Theorem~\ref{t:maximizer}}
\label{s:maximizer}

\begin{proof}
	First we show that deterministic policies maximize the information encoded in spikes, then we show that deterministic policies maximize reward per spike.
	
	\emph{Deterministic policies maximize effective information subject to $\pi(a_1)\leq \rho$.}
	Observe that
	\begin{align*}
		ei(\pi,a) & = \sum_{s\in\cS} \pi(s|a) \log \frac{\pi(s|a)}{P(s)} \\
		& = \sum_{s\in\cS} \pi(s|a) \log \frac{\pi(a|s)}{\pi(a)} \\
		& = \underbrace{-\log \pi(a)}_{\log \text{frequency of output }a} + \underbrace{\sum_{s\in\cS} \pi(s|a) \log \pi(a|s)}_{\text{stochasticity}}.
	\end{align*}
	The $\log$-frequency, or surprise, term is nonnegative and the stochasticity term is non-positive. It is easy to see that the stochasticity term is maximized at 0 if and only if output $a$ is chosen deterministically -- i.e. $\pi(a|s)$ is either 0 or 1 for all $s\in\cS$.

	\emph{Deterministic policies maximize reward per spike.}
	Suppose there are $N$ situations ordered according to their empirical reward, so that $\widehat{R}(s_1)< \cdots < \widehat{R}(s_N)$. It is clear that a deterministic policy spiking only for the $\rho\cdot N$ policies with highest empirical reward maximizes empirical reward after spiking.
\end{proof}

\section{Proof of Theorem~\ref{t:conc-inf}}
\label{s:mi-ei}

\begin{proof}
	Observe that
\begin{align*}
	P(s|a_0) & = \frac{P(s)-\pi(s|a_1)\cdot \pi(a_1)}{1-\pi(a_1)}\\
	& = P(s) + \pi(a_1)\big(P(s) - \pi(s|a_1)\big) + O\big(\pi(a_1)^2\big).
\end{align*}
Thus to first order in $\pi(a_1)$, $\pi(s|a_0)$ is of the form $p+\delta p$ where $\int\delta p =0$. We can then compute
\begin{align*}
	D\big[p+\delta p\,\big\|\, p\big] & = \int (p+\delta p)\log_2\frac{p+\delta p}{p}\\
	& = \alpha\int (p+\delta p)\frac{\delta p}{p}\left(1-\frac{\delta p}{2p}\right) +O\big((\delta p)^3\big)
	= \alpha\int \frac{(\delta p)^2}{2p}+ O\left((\delta p)^3\right),
\end{align*}
where $\alpha=\frac{1}{\ln 2}$.

Substituting $p=P(s)$ and $\delta p=\pi(a_1)\big(P(s) - \pi(s|a_1)\big)$ gives
\begin{align*}
	D\big[\pi(S|a_0)\,\big\|\, P(S)\big] & =
	D\big[p+\delta p\,\big\|\,p\big] +O\big(\pi(a_1)^2\big)\\
	& = \alpha\cdot 
	\pi(a_1)^2\int\frac{\big(P(s)-\pi(s|a_1)\big)}{2P(s)}
	+ O\big(\pi(a_1)^2\big).
\end{align*}

Thus $I(S;A) = \pi(a_1) D\big[\pi(S|a_1)\,\big\|\,P(S)\big]+O\big(\pi(a_1)^2\big)$.
\end{proof}

\section{Proof of Theorem~\ref{t:vcr}}
\label{s:Ockham}

Occam's razor can be paraphrased to say that the simplest hypothesis should be preferred. Suppose we have a setof hypotheses $\ccH$ with prior distribution $P(h)$ on $\ccH$. Let $-\log P(h)$ denote the complexity of hypothesis $h$. Let $L:\cX\times\ccH\rightarrow[0,b]$ be a loss function. Then

\begin{thm}[Occam's razor]\eod
	\label{t:occam}
	For any data generating distribution on $\cX$ and any prior distribution $P(h)$ over $\ccH$, with a probability greater than $1-\delta$ over drawing an i.i.d. sample from $\cX$ of size $T$, for all $h\in\ccH$:
	\begin{equation*}
		\left|L(h)-\widehat{L}(h)\right| \leq b\cdot \sqrt{\frac{-\log P(h)+\log\frac{2}{\delta}}{2T}}.
	\end{equation*}
\end{thm}

\begin{proof}
	See \cite{seldin:08}.
\end{proof}

Let $\ccH=\big\{\pi:\cS\rightarrow \cA\big\}$ denote the set of deterministic policies, where $\cA=\{a_0,a_1\}=\{0,1\}$. Define loss function
\begin{equation*}
	L:\Big(\cS\times \cN\Big)\times \ccH\longrightarrow\bR:(s,\nu)\times \pi\mapsto \reward(s,\pi(s),\nu).
\end{equation*}
Further, set probability distribution $P(s,\nu)=P(s)\cdot P(\nu|s)$ on $\cS\times\cN$. Theorem~\ref{t:occam} holds for any sampling distribution. In particular we may use the policy $\pi$ to restrict samples to situations that cause the neuron to spike to obtain $P(s,\nu|\pi(s)=a_1)=\pi(s|a_1)\cdot P(\nu|s)$.  It follows that
\begin{gather*}
	L(\pi) = \reward_\pi=\bE\Big[\reward(s,a_1,\nu)\,\Big|\,\pi(s|a_1)\cdot P(\nu|s)\Big]
	\text{ is the expected reward and}\\
	\widehat{L}(\pi) = \widehat{\reward}_\pi=\frac{1}{T_1}\sum_{\{(s_t,\nu_t)| \pi(s_t)=a_1\}} \reward\big(s_t,a_1,\nu_t\big)
	\text{ is the empirical reward,}
\end{gather*}
where $T_1$ is the number of spike produced by the neuron during $[1,T]$.

\begin{thmtag}{\ref{t:vcr}.}{}\eod
	\emph{Let $\ccH\supset\ccH_k=\big\{\pi:\cS\rightarrow\cA\text{ s.t. }|\pi^{-1}(a_1)|=k\big\}$ denote policies that spike for exactly $k$ situations.
	Given the setup above, with probability at least  $1-\delta$,
	\begin{equation*}
		{\left|\reward_\pi-\widehat{\reward}_\pi\right|} \leq
		b\cdot\sqrt{|\cS|\cdot\frac{ei_u(\pi,a_1)+1}{2T_1\cdot e^{ei_u(\pi,a_1)}}+\frac{\log\frac{2}{\delta}}{2T_1}}
	\end{equation*}
}
\end{thmtag}

\begin{proof}
	Let $N=|\cS|$ denote the number of possible situations. We put the uniform prior on $\ccH_k$, so $P(\pi)=\frac{1}{{N\choose k}}$. By Stirling's approximation, $\log{N\choose k}\leq k\log \left(\frac{N\cdot e}{k}\right)$, and it follows that
	\begin{equation*}
		-\log P(\pi) =  \log {N\choose k}\leq k\log\frac{N\cdot e}{k}=N\cdot\frac{k}{N}\cdot \left(\log\frac{N}{k}+1\right).
	\end{equation*}	
	The theorem follows since $\pi_u(a_1)=\frac{k}{N}$ and $ei_u(\pi,a_1)=\log \frac{N}{k}$.
\end{proof}


\begin{figure*}[t]
	\centering
	\subfigure[Effective information from spikes]{\includegraphics[width=0.85\textwidth]{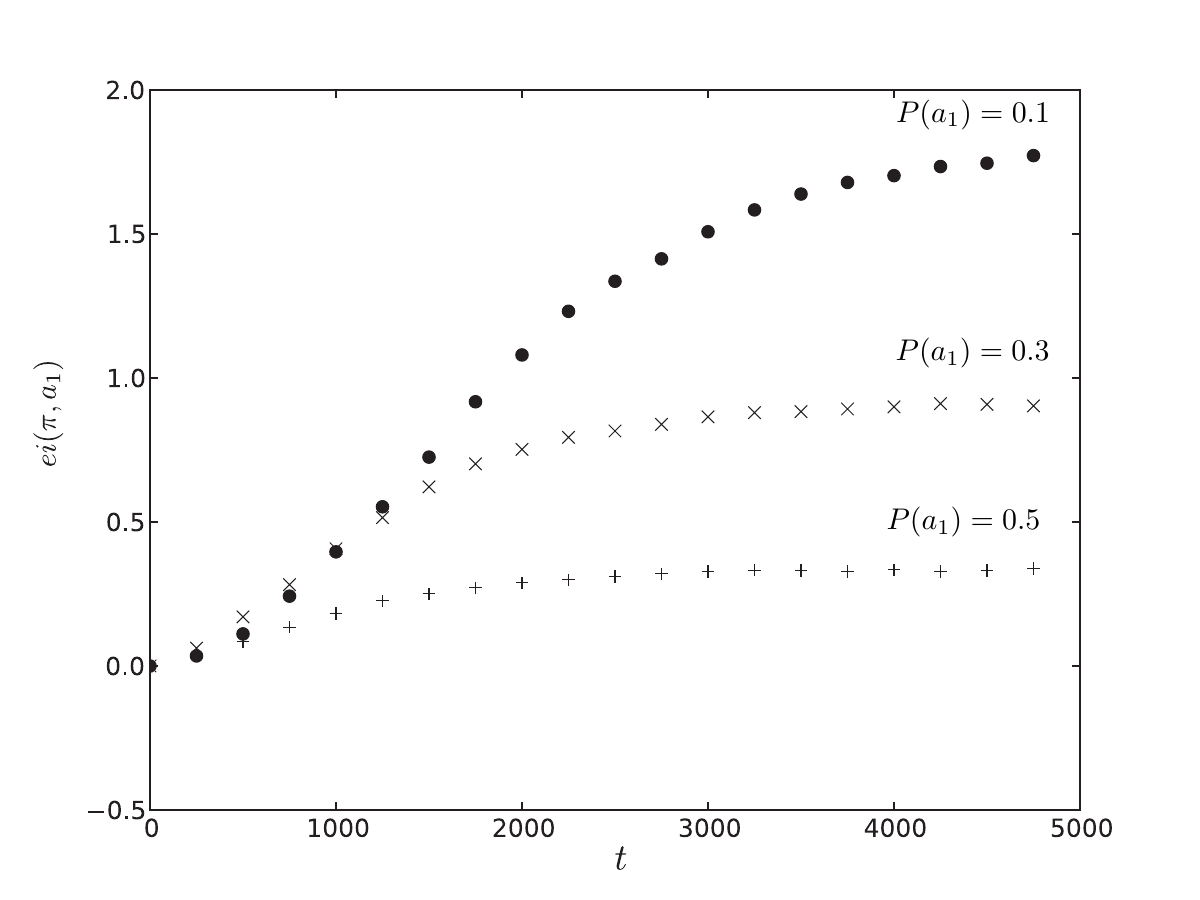}}
	\subfigure[Empirical reward after spiking]{\includegraphics[width=0.85\textwidth]{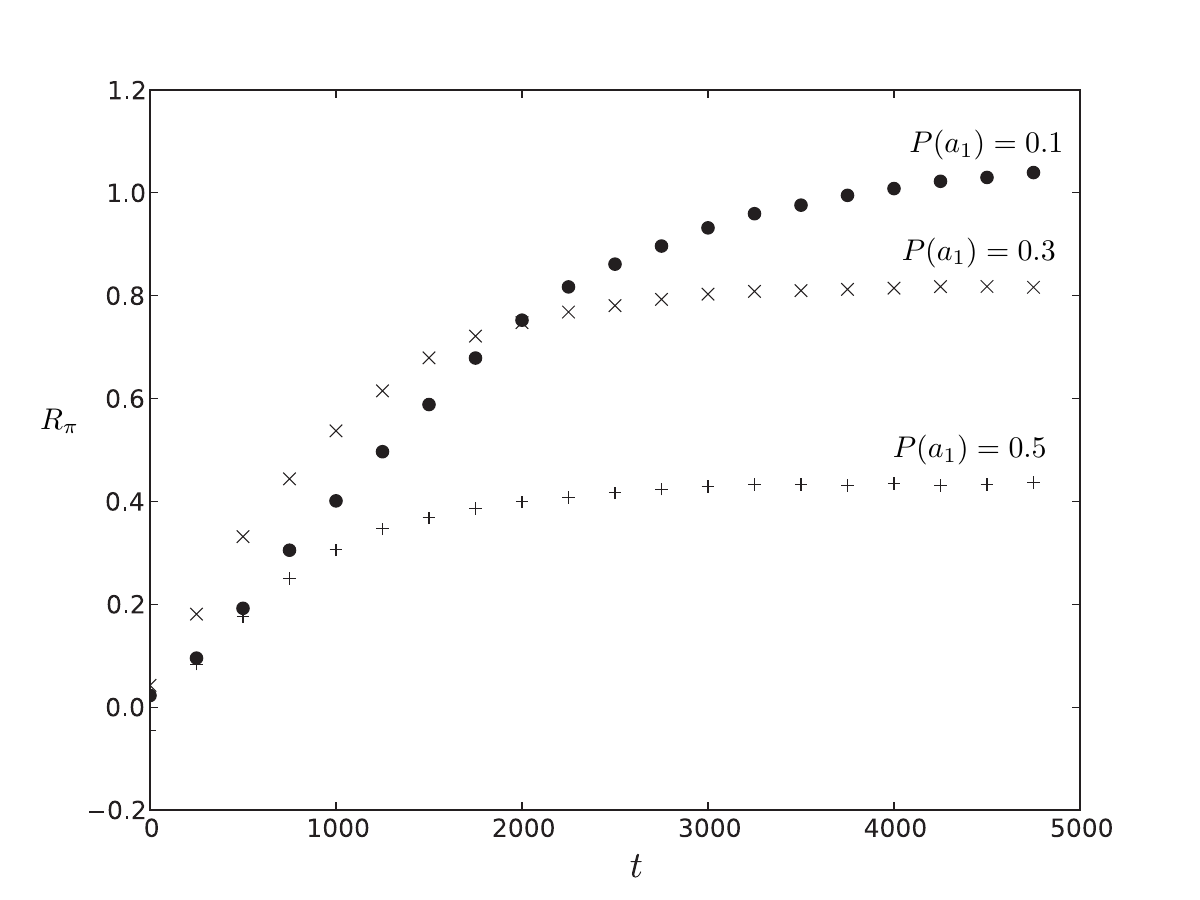}}
 	\caption{
	\textbf{Metabolically constrained $Q$-learning.}
	As neurons learn, effective information and empirical reward increase in qualitatively the same way. Tighter metabolic constraints yield both higher effective information and greater rewards.
	}
	\label{f:ei-ev}
\end{figure*}

\begin{figure}
	\begin{center}
		\includegraphics[width=0.85\textwidth]{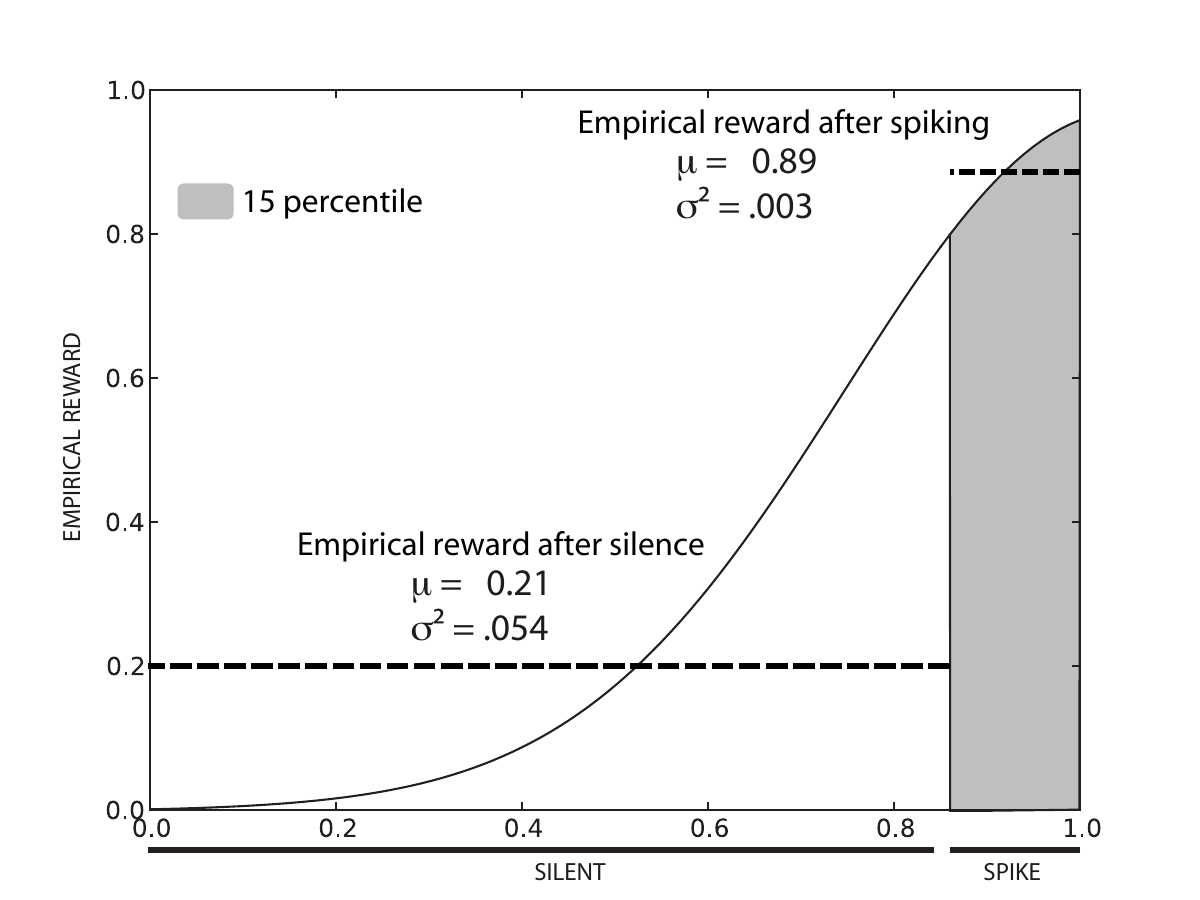}		
	\end{center}
 	\caption{
	\textbf{Concentration of reward.}
	The $x$-axis lists situations, ranked by the empirical reward the neuron would receive if it spiked. Situations are grouped into two categories: the top 15\%, which cause spikes, and the rest, which do not. The average and variance of the empirical reward in each category is displayed.
	}
	\label{f:bfo}
\end{figure}

\begin{figure}
	\begin{center}
		\includegraphics[width=0.85\textwidth]{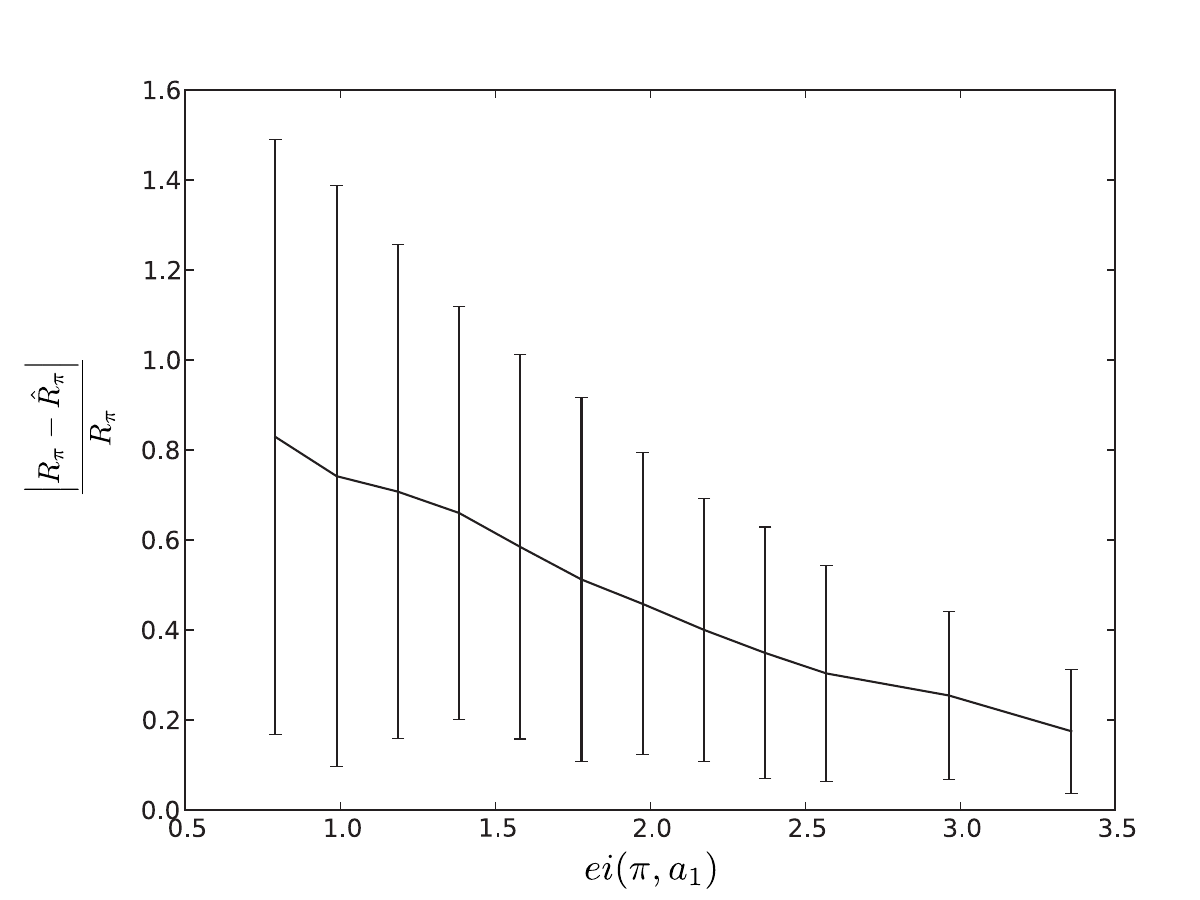}
	\end{center}
 	\caption{
	\textbf{Empirical versus expected reward.} 
	The normalized difference between expected and empirical reward, plotted against effective information.
	}
	\label{f:ei-bound}
\end{figure}

\end{document}